\documentclass[12pt]{article}
\usepackage{amsmath}
\usepackage{amssymb}
\usepackage{theorem}

\numberwithin{equation}{section}

\newtheorem{thm}{Theorem}[section]
\newtheorem{lemma}[thm]{Lemma}
\newtheorem{prop}[thm]{Proposition}
\newtheorem{cor}[thm]{Corollary}
{\theorembodyfont{\rmfamily}
\newtheorem{defn}[thm]{Definition}
\newtheorem{example}[thm]{Example}

\newtheorem{rmk}[thm]{Remark}
}

\newcommand{\qed}{\hfill \mbox{\raggedright \rule{.07in}{.1in}}}
 
\newenvironment{proof}{\vspace{1ex}\noindent{\bf
Proof}\hspace{0.5em}}{\hfill\qed\vspace{1ex}}

\newcommand{\R}{{\mathbb R}}
\newcommand{\T}{{\mathbb T}}
\newcommand{\C}{{\mathbb C}}
\newcommand{\Z}{{\mathbb Z}}
\newcommand{\N}{{\mathbb N}}
\newcommand{\var}{\operatorname{var}}
\newcommand{\dist}{\operatorname{dist}}

\newcommand{\cov}{\operatorname{cov}}

\title{On the Validity of the $0$-$1$ Test for Chaos}

\author{
Georg A. Gottwald 
\\ School of Maths and Stats
\\ University of Sydney
\\ NSW 2006, Australia
\and
Ian Melbourne 
\\ Department of Maths 
\\ University of Surrey
\\ Guildford GU2 7XH, UK
}

\date{}

\begin{document}

\maketitle

 \begin{abstract}
 In this paper, we present a theoretical justification of the $0$--$1$
 test for chaos.    In particular, we show that with probability one, the
 test yields $0$ for periodic and quasiperiodic dynamics, and $1$ for
 sufficiently chaotic dynamics.
 \end{abstract}

\section{Introduction} 
\label{sec-intro}

In~\cite{GM04}, we introduced a new method of detecting chaos in
deterministic dynamical system in the form of a binary test.  The
method applies directly to the time series data and does not require
phase space reconstruction.  As explained in~\cite{GM04}, with
probability one the test gives the output $K=0$ for quasiperiodic
dynamics and $K=1$ for sufficiently chaotic dynamics.

In~\cite{GM05}, we proposed a simplified version of the test that is
more effective for systems with a moderate amount of noise.  The
effectiveness of the new method was demonstrated for
higher-dimensional systems in~\cite{GM05} and for experimental
data~\cite{FGMW07}.

The main aim of this paper is to put the simplified version of the
test on a rigorous footing, going far beyond the results indicated
in~\cite{GM04} for the original test.  In addition, our analysis of
the test leads to a significant improvement which was used in our
paper~\cite{GMapp} detailing the implementation of the test.

We first recall the simplified form of the test proposed
in~\cite{GM05}.  Let $f:X\to X$ be a map with invariant ergodic
probability measure $\mu$.  Let $v:X\to\R$ be a scalar
square-integrable observable.  Choose $c\in(0,2\pi)$, $x\in X$, and
define
\begin{align} \label{eq-p}
p_c(n)=\sum_{j=0}^{n-1}e^{ijc}v(f^jx).
\end{align}
Next, define the mean-square displacement
\begin{align} \label{eq-M}
M_c(n)=\lim_{N\to\infty}\frac1N\sum_{j=1}^N|p_c(j+n)-p_c(j)|^2.
\end{align}
Finally, let
\begin{align} \label{eq-K}
K_c=\lim_{N\to\infty}\frac{\log M_c(n)}{\log n}.
\end{align}
The claim in~\cite{GM05}, substantiated in this paper, is that
typically (i) the limit $K_c$ exists, (ii) $K_c\in\{0,1\}$, and (iii)
$K_c=0$ signifies regular dynamics while $K_c=1$ signifies chaotic
dynamics.

\begin{rmk}   \label{rmk-K}   (a)
The definition of $p_c(n)$ in~\eqref{eq-p} is slightly different from in~\cite{FGMW07, GM05,GMapp} where
$p_c(n)=\sum_{j=0}^{n-1}\cos(jc)v(f^jx)$.  In the current paper it is
natural to simplify analytic calculations rather than numerical
computations, but apart from that the methods are equivalent.

(b) For fixed $c$, it follows from the ergodic theorem that the limit
$M_c(n)$ in~\eqref{eq-M} exists for almost every initial condition $x$
and the limit is independent of $x$.  The common limit is
\[
M_c(n)=\int_X |p_c(n)|^2d\mu=\|p_c(n)\|_2^2.
\]
To see this, compute that $p_c(j+n)-p_c(j)=e^{ijc}p_c(n)\circ f^j$,
and so $M_c(n)=\lim_{N\to\infty}\frac1N\sum_{j=0}^{N-1}|p_c(n)|^2\circ
f^j$ which converges to the space average $\int_X |p_c(n)|^2 d\mu$
almost everywhere.

(c) Strictly speaking, the limit $K_c$ in~\eqref{eq-K} need not be
well-defined.  Of course, $K_c^+=\limsup_{N\to\infty}\log M_c(n)/\log
n$ is well-defined, and it follows from Proposition~\ref{prop-K+} that
$K_c^+\in[0,2]$ for all $c$.  (In the case of periodic dynamics,
$K_c=2$ for isolated values of $c$.)
\end{rmk}

\begin{example}   \label{ex-logistic}
Consider the logistic map $f:[0,1]\to[0,1]$ given by $f(x)=ax(1-x)$
for $0\le a\le 4$.  This family of maps is particularly
well-understood~\cite{Lyubich02,AvilaMoreira05}: we can
decompose the parameter interval according to
$[0,4]=\mathcal{P}\cup\mathcal{C}\cup\mathcal{N}$ where $\mathcal{N}$
has Lebesgue measure zero and the asymptotic dynamics consists of a
periodic attractor (of period $q\ge1$) for $a\in\mathcal{P}$ and a
strongly chaotic attractor consisting of $q\ge1$ disjoint intervals
for $a\in\mathcal{C}$ (satisfying the Collet-Eckmann condition).

We obtain the following result:
\begin{prop}  \label{prop-logistic}
Let $v:[0,1]\to\R$ be H\"older.
Let $a\in\mathcal{P}\cup\mathcal{C}$ and define $q$ as above.
\begin{itemize}
\item[(a)]  If $a\in\mathcal{P}$, then $K_c=0$ for all $c\neq 2\pi j/q$.
\item[(b)]  If $a\in\mathcal{C}$, then $K_c=1$ for all $c\neq 2\pi j/q$
unless $v$ is infinitely degenerate\footnote{Lying in a closed
subspace of infinite codimension in the space of H\"older functions}.
\end{itemize}
\end{prop}
Hence, the test succeeds with probability one for logistic map
dynamics.
\end{example}

Part (a) holds for general periodic dynamics (and all continuous
observables).  In Section~\ref{sec-regular}, we prove that the test
yields $K_c=0$, for almost all $c$, for quasiperiodic dynamics,
provided we make smoothness assumptions on $v$.  This justifies our
claim that $K_c=0$ for regular dynamics.

The chaotic case is discussed extensively in
Section~\ref{sec-chaotic}.  In particular we obtain \mbox{$K_c=1$}
under various assumptions:
\begin{description}
\item[(i) Positivity of power spectra]
\item[(ii) Exponential decay of autocorrelations] 
\item[(iii) Summable decay of autocorrelations plus hyperbolicity]
\end{description}
(In fact, (ii) and (iii) are sufficient conditions for (i).)

In many situations, including the logistic map with
$a\in\mathcal{C}$, it is necessary to consider $f^q$ instead of $f$,
and autocorrelations decay only up to a finite cycle (of length $q$).
As shown in Section~\ref{sec-chaotic}, criteria (ii) and (iii)
generalise to this situation.

\paragraph{Summable decay without hyperbolicity assumptions}
Without making hyperbolicity assumptions, we have no definitive
results when autocorrelations decay subexponentially.  However, there
is some partial information discussed in Section~\ref{sec-summable}.
If the autocorrelation function 
is summable, then the power spectrum $S(c)$ exists for all
$c\in(0,2\pi)$ by the Wiener-Khintchine Theorem~\cite{VanKampen},
implying that
\begin{align*}
M_c(n)=S(c) n +o(n).
\end{align*}
Under slightly stronger assumptions on the decay rate
\begin{align*}
M_c(n)=S(c) n + O(1).
\end{align*}
In the former case, $K^+_c=\limsup_{n\to\infty}\log M_c(n)/\log
n\in[0,1]$.  In the latter case, $K_c$ exists and takes the value $0$
or $1$ depending on where $S(c)=0$ or $S(c)>0$ (but see
Remark~\ref{rmk-exp}).

Again, we obtain similar results if autocorrelations are summable up to 
a finite cycle.

\paragraph{Improved diagnostic in the test for chaos}
The $o(n)$ and $O(1)$ terms above are nonuniform in $c$ but
in Section~\ref{sec-summable} we show
that the source of nonuniformity is easily dealt with.
Define 
\[
D_c(n)=M_c(n)- (Ev)^2\frac{1-\cos nc}{1-\cos c}.   
\]
Here $Ev=\int_X v\, d\mu$ denotes expectation with respect to $\mu$.
Under the above conditions we obtain $D_c(n)=S(c) n +o(n)$ 
(hence $K^+_c=\limsup_{n\to\infty}\log M_c(n)/\log n\in[0,1]$)
for summable autocorrelation functions
and $D_c(n)=S(c) n + O(1)$ (hence 
$K_c$ takes the values either $0$ or $1$) under slightly
stronger conditions on the decay of the autocorrelation function
as before, but the $o(n)$ and $O(1)$ terms are now uniform in $c$ (see
Section~\ref{sec-summable}).  In~\cite{GMapp}, we proposed using
$D_c(n)$ instead of $M_c(n)$ in the numerical implementation of the
$0$--$1$ test, and demonstrated the improved performance of the test.

\paragraph{Nonsummable decay}
The summability condition in the Wiener-Khintchine Theorem can be
weakened considerably.  For example, if autocorrelations decay at a
square summable rate 
(including $k^{-d}$ for any $d>\frac12$), then the
power spectrum exists almost everywhere and so $M_c(n)=S(c)n+o(n)$ for
almost every $c$.  (In this generality there is no uniformity in the
error term for $D_c(n$.)  This and related results is discussed in
Section~\ref{sec-nonsum}.

\paragraph{Correlation method}
Our emphasis in this paper is on understanding the properties
of the limit $K_c$ as defined in~\eqref{eq-K}.  However,
in~\cite{GMapp}, we proposed
computing $K_c$ as the correlation of the mean-square
displacement $M_c(n)$ (or $D_c(n)$) with $n$.
The advantages of this approach were demonstrated in~\cite{GMapp}.
In Section~\ref{sec-corr}, we verify that the theoretical
value of $K_c$ remains $0$ for regular dynamics and $1$ for chaotic dynamics.

\paragraph{}
The paper concludes with a discussion section (Section~\ref{sec-potato}).
We end the introduction by proving that $K_c^+\in[0,2]$ as claimed in
Remark~\ref{rmk-K}(c).

\begin{prop} \label{prop-K+}
Let $K^+_c=\limsup_{n\to\infty}\log M_c(n)/\log n$.
If $v$ is not identically zero, then $K^+(c)\in[0,2]$ for all $c$.
\end{prop}

\begin{proof}
By definition, $\|p_c(n)\|_2\le n\|v\|_2$ so that 
$0\le M_c(n)\le n^2\|v\|_2^2$.   Hence $K_c^+\le 2$.

To prove the lower bound, we use the fact that $\|v\|_2>0$.
It suffices to show
that $\limsup_{n\to\infty}M_c(n)>0$ for each fixed $c$.
Observe that $p_c(n+1)=e^{inc}v\circ f^n+p_c(n)$ so that 
\[
\|p_c(n+1)\|_2 \ge \|e^{inc}v\circ f^n\|_2-\|p_c(n)\|_2 = \|v\|_2-\|p_c(n)\|_2.
\]
Hence $0<\|v\|_2\le \|p_c(n)\|_2+\|p_c(n+1)\|_2$.
It follows that $\|p_c(n)\|_2\not\to0$, and so 
$\limsup_{n\to\infty}M_c(n)>0$ as required.
\end{proof}

\section{The case of regular dynamics}
\label{sec-regular}

Part (a) of Proposition~\ref{prop-logistic} is a simple direct calculation.   If $f:X\to X$
is a map with a periodic orbit of period $q$ and $v:X\to\R$ is continuous, then we obtain
$K_c=0$ for all $c\neq 2\pi j/q$.    
(For isolated resonant values $c=2\pi j/q$ a simple argument using the
Fourier series for $v$ shows that typically $p_c(n)$ will grow
linearly implying $K_c=2$.)
In the case of quasiperiodic dynamics, we require additional
smoothness assumptions on the observable $v$.  The test then succeeds
with probability one.

\begin{thm}   \label{thm-qp}
Suppose that $X=\T^m=\R^m/(2\pi\Z)^m$ and that $f:X\to X$ is given by
$f(x)=x+\omega\bmod 2\pi$.  If $v:X\to\R$ is $C^r$ with $r>m$, then
$K_c=0$ for almost every $c\in[0,2\pi]$.
\end{thm}

\begin{proof}
Write $v:X\to\R$ as a $m$-dimensional Fourier
series $v(x)=\sum_{\ell\in\Z^m}v_\ell e^{i\ell\cdot x}$ where
$v_{-\ell}=\bar v_\ell$.  Then
\begin{align}
\nonumber M_c(n) & =\int_X|\sum_{j=0}^{n-1} e^{ijc}v\circ f^j|^2\,dx
= \sum_{p,q=0}^{n-1} e^{i(p-q)c}\int_X v\circ f^p\,v\circ f^q\,dx \\
& = \sum_{k=-(n-1)}^{n-1}(n-|k|)e^{ikc}\int_X v\circ f^{|k|}\,v\,dx 
 =s_1+\dots+s_n,
\label{eq-Msn}
\end{align}
where 
\[
s_m=\sum_{j=-(m-1)}^{m-1}e^{ijc}\int_X v\circ f^{|j|}\,v\,dx.
\]
We show that $M_c(n)$ is bounded (as a function of $n$) for almost all $c$.
Compute (formally) that
\begin{align} \nonumber
s_n & =\sum_{j=-(n-1)}^{n-1}e^{ijc}\sum_{\ell} |v_\ell|^2 
e^{ij(\ell\cdot\omega)} \\
& = \sum_{\ell}|v_\ell|^2 (e^{i(c+\ell\cdot\omega)}-1)^{-1}(e^{in(c+\ell\cdot\omega)}-e^{-i(n-1)(c+\ell\cdot\omega)}).  \label{eq-formal}
\end{align}
Hence
\begin{align} \label{eq-formal2}
s_1+\dots+s_n=\sum_{\ell}|v_\ell|^2\frac{1-\cos n(c+\ell\cdot\omega)}{1-\cos(c+\ell\cdot\omega)}.
\end{align}

It remains to show that the
series~\eqref{eq-formal},~\eqref{eq-formal2} converge.  
We may ignore the $\ell=0$ term in these series
(these terms are obviously bounded in $n$).
The smoothness assumption on $v$ implies that $|v_\ell|=O(|\ell|^{-r})$.
Let $\epsilon>0$.  For almost every $c\in(0,2\pi)$ there is a constant
$d_0>0$ such that 
\begin{align}   \label{eq-Dio}
\dist|c+\ell\cdot\omega,2\pi\Z|\ge
d_0|\ell|^{-(m+\epsilon)},
\end{align}
for all $\ell\in\Z^m-\{0\}$ (cf.~\cite{NMA01}).  
Hence $|e^{i(c+\ell\cdot\omega)}-1|\ge d_1|\ell|^{-(m+\epsilon)}$ and so
\begin{align*}
\sum_{\ell\in\Z^m-\{0\}}|v_\ell|^2|e^{i(c+\ell\cdot\omega)}-1|^{-1}
& \le \sum_{k=1}^\infty\sum_{|\ell|=k}k^{-2r}d_1^{-1}k^{m+\epsilon}
 \le C\sum_{k=1}^\infty k^{m-1}k^{-2r}k^{m+\epsilon} \\ &
=C\sum_{k=1}^\infty k^{-(1+2(r-m-\epsilon/2))}<\infty,
\end{align*}
provided we choose $\epsilon>0$ so small that $r>m+\epsilon/2$.
This shows that~\eqref{eq-formal} converges and is bounded independent
of $n$, and similarly for~\eqref{eq-formal2}.
\end{proof}

\begin{rmk}   The extra smoothness of $v$ is required to circumvent the small divisor
problems associated with quasiperiodic dynamics. We also require a
Diophantine condition on $c$, satisfied by almost every
$c\in[0,2\pi]$.  However, there is no restriction on $\omega$.
\end{rmk}

\section{The case of chaotic  dynamics}
\label{sec-chaotic}

It is our intention to show that $K_c=1$ for all almost all $c$ (and reasonable
observables $v$) for sufficiently chaotic dynamical systems.    We proceed along three
distinct but related avenues, all of which extend
Example~\ref{ex-logistic}
of the logistic map:
(i)  positivity of power spectra;
(ii)  decay of autocorrelation functions;
(iii) hyperbolicity of the dynamical system.

Recall that for a square-integrable observable $v:X\to\R$ the {\em
power spectrum} $S:[0,2\pi]\to[0,\infty)$ is defined (assuming it
exists) to be the square of the Fourier amplitudes of $v\circ f^j$ per unit
time\footnote{Often $e^{ij\omega}$ is replaced by $e^{2\pi
ij\omega/n}$ in the literature, but this is just a rescaling of the
domain.}, and is given by
\[
S(c)=\lim_{n\to\infty}\frac1n \int_X |\sum_{j=0}^{n-1}e^{ijc}v\circ f^j|^2\,d\mu
=\lim_{n\to\infty}\frac1n M_c(n).
\]
In other words, $M_c(n)=S(c)n+o(n)$. The following result is immediate:
\begin{prop}  \label{prop-power}
Let $c\in[0,2\pi]$.  Suppose that $S(c)$ is well-defined and strictly
positive.  Then $K_c=1$. \qed
\end{prop}
In particular, if the power spectrum is well-defined and positive
almost everywhere, then we obtain $K_c=1$ with probability one.

Next, we consider the {\em autocorrelation function} $\rho:\N\to\R$ given by
\[
\rho(k)=\int_X v\circ f^k\,v\,d\mu-\Bigl(\int_X v\,d\mu\Bigr)^2.
\]
This is well-defined for all $L^2$ observables $v$.

If $\rho(k)$ is summable (i.e.\ $\sum_{k=0}^\infty|\rho(k)|<\infty$),
then it follows from the Wiener-Khintchine theorem~\cite{VanKampen}
that for $c\in(0,2\pi)$,
\[
S(c)=\sum_{k=-\infty}^\infty e^{ikc}\rho(|k|).
\]
Note that the right-hand-side defines a continuous function on
$[0,2\pi]$.

\begin{prop}   \label{prop-exp}
Suppose that $v:X\to\R$ lies in $L^2(X)$ and $v$ is not constant
(almost everywhere).  If the autocorrelation function $\rho(k)$ decays
exponentially\footnote{There exist constants $C\ge1$, $\tau\in(0,1)$
such that $|\rho(k)|\le C\tau^k$}, then $K_c=1$ except for at most
finitely many choices of $c\in[0,2\pi]$.
\end{prop}

\begin{proof}
Since $\rho(k)$ decays exponentially,
$g(c)=\sum_{k=-\infty}^\infty e^{ikc}\rho(|k|)$ is analytic on $[0,2\pi]$.  
Since $v$ is not constant, 
$\rho(0)=\int_X v^2\,d\mu-(\int_X v\,d\mu)^2=\int_X(v-\int_X v)^2\,d\mu>0$, 
and hence $g$ is not the
zero function.  By analyticity, $S(c)=g(c)>0$ except for at most
finitely many values of $c$ and hence $K_c=1$ except for these values of $c$.
\end{proof}

\paragraph{Decay of autocorrelations up to a finite cycle}
Recall that $f:X\to X$ is {\em mixing} if $\lim_{k\to\infty}\rho(k)
\to0$ for every $L^2$ observable $v:X\to\R$.
The system is {\em mixing up to a finite cycle} (of length $q\ge1$) if
$X=X_1\cup\dots\cup X_q$ where $f(X_j)\subset X_{j+1}$ (computing
indices $\bmod\, q$) and $f^q:X_j\to X_j$ is mixing (with respect to
$\mu_q=q\mu|X_j$) for each $j=1,\dots,q$.

If $q\ge2$, then decay of autocorrelations holds only for
degenerate observables.   The natural property to require is exponential 
decay for $f^q$.
Given an $L^2$ observable $v:X\to\R$, define 
for $j=1,\dots,q$ and $m=0,\dots,q-1$, 
\[
\rho_{v\circ f^m,v,X_j}(kq)=\int_{X_j} v\circ f^m\circ f^{kq}\, v\,d\mu_j
-\int_{X_j}v\circ f^m\,d\mu_j \int_{X_j}v\,d\mu_j.
\]

\begin{defn} \label{def-expcycle}
The autocorrelations of $v$ are {\em summable up to a $q$ cycle} if 
for each $j=1,\dots,q$ and $m=0,\dots,q-1$, the series 
$\sum_{k=0}^\infty|\rho_{v\circ f^m,v,X_j}(kq)|$ is convergent.

The autocorrelations of $v$ {\em decay exponentially up to a $q$ cycle} if 
$\rho_{v\circ f^m,v,X_j}(kq)$ decays exponentially as $k\to\infty$
for each $j=1,\dots,q$ and $m=0,\dots,q-1$.
\end{defn}

\begin{thm} \label{thm-expcycle}
If the autocorrelations of $v$ are summable up to a $q$ cycle, then 
\[
S(c)=\sum_{r=-\infty}^\infty e^{irc}g_r, \enspace\text{for all  
$c\neq 2\pi j/q$},
\]
where writing $r=kq+m$ with $k\in\Z$ and $m\in\{0,1,\dots,q-1\}$,
\[
g_r=g_{kq+m}=\sum_{j=1}^q \rho_{v\circ f^m,v,X_j}(kq).
\]
\end{thm}

\begin{proof}
Define
\[
v_{q,c}=\sum_{\ell=0}^{q-1}e^{i\ell c}v\circ f^\ell, \quad
\rho_{q,c}(k)=\frac1q\sum_{j=1}^q\Bigl(\int_{X_j} v_{q,c}\circ f^{kq}\bar v_{q,c}\,d\mu_j-
\Bigl|\int_{X_j} v_{q,c}\,d\mu_j\Bigr|^2\Bigr).
\]
By~\cite[Theorem~A.2]{MG08},
\[
S(c)=\sum_{k=-\infty}^\infty e^{ikqc}\rho_{q,c}(k),
\]
for $c\neq 2\pi j/q$.
Compute that
\begin{align*}
\rho_{q,c}(k) & = \frac1q\sum_{j=1}^q \Bigl(\int_{X_j}v_{q,c}\circ f^{kq}\,\bar v_{q,c}\,d\mu_j - \Bigl|\int_{X_j}v_{q,c}\,d\mu_j\Bigr|^2\Bigr) \\
& =\frac1q \sum_{j=1}^q \sum_{\ell,\ell'=0}^{q-1} e^{i(\ell-\ell')c}\Bigl(\int_{X_j}
v\circ f^\ell\circ f^{kq}\, v\circ f^{\ell'}\,d\mu_j -
\int_{X_j}v\circ f^\ell\,d\mu_j \int_{X_j}v\circ f^{\ell'}\,d\mu_j \Bigr)
\\
& =\frac1q \sum_{\ell,\ell'=0}^{q-1}e^{i(\ell-\ell')c}\sum_{j=1}^q \Bigl(
\int_{X_{j+\ell'}} v\circ f^{\ell-\ell'}\circ f^{kq}\, v\,d\mu_{j+\ell'}
-
\int_{X_{j+\ell'}}v\circ f^{\ell-\ell'}\,d\mu_{j+\ell'} \int_{X_{j+\ell'}}v\,d\mu_{j+\ell'}\Bigr)
\\
& =\frac1q \sum_{\ell,\ell'=0}^{q-1}e^{i(\ell-\ell')c} \sum_{j=1}^q\Bigl(
\int_{X_j} v\circ f^{\ell-\ell'}\circ f^{kq}\, v\,d\mu_j 
-\int_{X_j}v\circ f^{\ell-\ell'}\,d\mu_j \int_{X_j}v\,d\mu_j\Bigr)
\\
& =\frac1q \sum_{j=1}^q\sum_{s=-(q-1)}^{q-1} (q-|s|)e^{isc}\Bigl(\int_{X_j}
v\circ f^s\circ f^{kq}\, v\,d\mu_j
-\int_{X_{j}}v\circ f^s\,d\mu_{j}\int_{X_j}v\,d\mu_j\Bigr) \\
& =\frac1q \sum_{j=1}^q\sum_{s=-(q-1)}^{q-1} (q-|s|)e^{isc}\rho_{v\circ f^s,v,X_j}(kq).
\end{align*}

For $r=kq+m$ we obtain at most two nonzero contributions to $g_r$, namely
$s=m$ in $\rho_{q,c}(k)$ and $s=-(q-m)$ in $\rho_{q,c}(k+1)$.   Hence
\begin{align*}
g_r &=\frac1q\sum_{j=1}^q\Bigl((q-m)\rho_{v\circ f^m,v,X_j}(kq)+m\rho_{v\circ f^{-(q-m)},v,X_j}(kq+q)\Bigr) 
  = \sum_{j=1}^q \rho_{v\circ f^m,v,X_j}(kq)
\end{align*}
as required.
\end{proof}

\begin{cor} \label{cor-expcycle}
If the autocorrelations of $v$ are summable up to a $q$ cycle, then 
$S(c)$ exists and is continuous except for removable
singularities at $c= 2\pi j/q$.  

If the autocorrelations of $v$ decay exponentially up to a $q$ cycle,
then $S(c)$ is analytic except for removable singularities at $c= 2\pi j/q$.
If moreover $v|X_j$ is not constant (almost everywhere) for at least one $j$,
then $K_c=1$ except for at most finitely many values of $c$.
\end{cor}

\begin{proof}
The statements about continuity and analyticity are immediate
from Theorem~\ref{thm-expcycle}.
In particular, if there is exponential decay up
to a $q$ cycle, then the function $g(c)=\sum_{r=-\infty}^\infty e^{irc}g_r$ is
analytic and hence nonzero except at finitely many points provided
$g_0\neq0$.
If on the other hand, $g_0=0$, then
\[
0=g_0=\sum_{j=1}^q \rho_{v,v,X_j}(0)=\sum_{j=1}^q \int_{X_j}v^2\,d\mu_j
-\Bigl(\int_{X_j}v\,d\mu_j\Bigr)^2
=\sum_{j=1}^q \var(v|X_j),
\]
so $\var(v|X_j)=0$, and hence $v|X_j$ is constant, for each $j$.
\end{proof}

\begin{rmk} \label{rmk-expcycle}   Definition~\ref{def-expcycle} and 
Theorem~\ref{thm-expcycle}
are significant improvements on the corresponding material
in~\cite[Appendix]{MG08}.
\end{rmk}

\begin{rmk}   \label{rmk-exp}
We have seen that exponential decay of autocorrelations (up to a $q$ cycle)
guarantees that $K_c=1$ with probability one.
Surprisingly, it seems nontrivial to weaken the exponential decay hypothesis.
The proof of Proposition~\ref{prop-exp}
relies crucially on analyticity of the power spectrum.
Even if we assume sufficiently rapid decay that
$g(c)=\sum_{k=-\infty}^\infty e^{ikc}\rho(|k|)$ is $C^\infty$, then 
we face the difficulty that the only restriction on the zero set of a 
$C^\infty$ function is that it is a closed set.

Suppose that summable decay of correlations holds for a large class of
observables $\mathcal{B}$ with the property that there is an interval
$I\subset(0,2\pi)$ such that the Fourier series $g(c)$ is identically
zero on $I$ for all $v\in\mathcal{B}$.  The proof of
Proposition~\ref{prop-exp} shows that for every nonconstant observable
$v\in\mathcal{B}$, there is an interval $J\subset(0,2\pi)$ on which
$g(c)>0$ on $J$.  For such examples, where the power spectrum vanishes
on an interval $I$ and is typically positive on an interval $J$, the
$0$--$1$ test is inconclusive: we obtain $K_c=1$ with
positive probability for nonconstant observables in $\mathcal{B}$,
but for all $v\in\mathcal{B}$ there is a positive probability
that $K_c\in[0,1)$ or that $K_c$ does not even exist.  
This situation seems highly
pathological, but we do not see how to rule this out.
\end{rmk}

\paragraph{Hyperbolicity}
We can overcome the unsatisfactory aspects of Remark~\ref{rmk-exp} by
assuming some hyperbolicity.  In the Collet-Eckmann case
($a\in\mathcal{C}$) for the logistic map, it is known that H\"older
observables enjoy exponential decay of correlations up to a finite
cycle, so we can apply Theorem~\ref{thm-expcycle}.
Alternatively,~\cite{MG08} shows that the power spectrum is bounded
away from zero for all $c\neq 2\pi j/q$, so we can apply
Proposition~\ref{prop-power}.  These comments apply to all maps in the
following classes:
\begin{itemize}

\item   Uniformly expanding maps; Uniformly hyperbolic (Axiom~A)
diffeomorphisms.
\item   Nonuniformly expanding/hyperbolic systems in the sense of
Young~\cite{Young98}, modelled by a Young tower with exponential
tails. These enjoy exponential decay of correlations (up to a finite
cycle) for H\"older observables.  This covers large classes of
dynamical systems, including H\'enon-like maps, logistic maps and more
generally multimodal maps satisfying Collet-Eckmann conditions, and
one-dimensional maps with Lorenz-like singularities~\cite{DHL}.
\end{itemize}

Young~\cite{Young99} weakens the decay rates assumed for tower models
for nonuniformly expanding/hyperbolic systems. H\"older observables
now have subexponential decay of correlations (up to a finite cycle).
Provided the decay rate is summable, the argument of~\cite{MG08} still
applies, and $S(c)$ is bounded away from zero (except for infinitely
degenerate observables).

\begin{example}
A prototypical family of examples is the Pomeau-Manneville
intermittency maps $f:[0,1]\to[0,1]$ given by $f(x)=\begin{cases}
x(1+2^\alpha x^\alpha); & 0\le x\le \frac12 \\ 2x-1; & \frac12\le
x\le1 \end{cases}$ where $\alpha\in[0,1)$ is a
parameter~\cite{PM,LSV}.  When $\alpha=0$ this is the doubling map
with exponential decay of correlations for H\"older observables, so
Proposition~\ref{prop-exp} applies.  For $\alpha>0$, let
$\beta=\frac{1}{\alpha}-1$.  Then decay of correlations for H\"older
observables is at the rate $n^{-\beta}$~\cite{Hu}.  By~\cite{MG08},
the power spectrum is bounded below in the summable case ($\beta>1$,
equivalently $\alpha<\frac12$) and so $K_c=1$ for all $c\in(0,2\pi)$.
\end{example}

\begin{rmk}  Numerical experiments for intermittency maps
indicate that (i) $K_c=1$ and (ii) the power spectrum exists and
is bounded below, even in the nonsummable case $\beta\le 1$,
equivalently $\alpha\in[\frac12,1)$.  It remains an interesting
problem to prove these statements.
By Corollary~\ref{cor-d>.5}, we are at least assured that the power
spectrum exists (and hence $K_c=1$ with positive probability) for
$\beta>\frac12$ ($\alpha<\frac23$).
\end{rmk}

\section{Summable decay without hyperbolicity}
\label{sec-summable}

Recall that if the autocorrelation function is summable
(i.e. $\sum_{k=1}^\infty|\rho(k)|<\infty$), then the power spectrum
$S(c)$ exists and is continuous for all $c\in(0,2\pi)$.   Indeed,
$S(c)=\sum_{k=-\infty}^\infty e^{ikc}\rho(|k|)$ on $(0,2\pi)$ by
the Wiener-Khintchine Theorem~\cite{VanKampen}. 
Hence $M_c(n)=S(c)n+o(n)$.   As noted in Remark~\ref{rmk-exp},
$S(c)$ is not identically zero for nonconstant observables, and hence there
is an interval of values of $c$ for which $S(c)>0$.  
In particular, $K_c=1$ with positive probability.

In this section, we discuss the error term $o(n)$ in more detail.
As mentioned in the introduction, this leads to the improved diagnostic
$D_c(n)$ for chaos used in~\cite{GMapp}.

We begin with a formal calculation to express the mean square
displacement as follows:
\begin{prop} \label{prop-calc}
$M_c(n)= 
\sum_{k=-n}^n(n-|k|)e^{ikc}\rho(|k|)+ (Ev)^2\,\frac{1-\cos nc}{1-\cos c}$.
\end{prop}

\begin{proof}
First, note that
\begin{align*}
M_c(n)=\int_X |p_c(n)|^2\, d\mu
& =\sum_{p,q=0}^{n-1}e^{i(p-q)c}\int_X v\circ f^p\,v\circ f^q\, d\mu \\ 
& =\sum_{p,q=0}^{n-1}e^{i(p-q)c}\int_X v\circ f^{|p-q|}\,v\, d\mu \\ 
& =\sum_{p,q=0}^{n-1}e^{i(p-q)c}(\rho(|p-q|)+(Ev)^2) \\
& =\sum_{k=-n}^n(n-|k|)e^{ikc}\rho(|k|)+
(Ev)^2\sum_{p=0}^{n-1}e^{ipc} \sum_{q=0}^{n-1}e^{-iqc}
\end{align*}
Finally
\begin{align*}
\sum_{p=0}^{n-1}e^{ipc} \sum_{q=0}^{n-1}e^{-iqc}
=\frac{1-e^{inc}}{1-e^{ic}}\,\frac{1-e^{-inc}}{1-e^{-ic}}
=\frac{1-\cos nc}{1-\cos c}.
\end{align*}

\vspace{-4ex}
\end{proof}

The second term in the expression for $M_c(n)$ is bounded in $n$ for fixed
$c$, but is nonuniform in $c$.   Since the term is explicit, it is convenient
to remove it. (As demonstrated in~\cite{GMapp}, this is also greatly
advantageous for the numerical implementation of the test.)
Hence we define
\[
D_c(n)=M_c(n)-(Ev)^2(1-\cos nc)/(1-\cos c)\; .
\]
Then $M_c(n)=D_c(n)+O(1)$, so it suffices to work with $D_c(n)$ from
now on. By Proposition~\ref{prop-calc},
\begin{align} \label{eq-D}
D_c(n)=\sum_{k=-n}^n(n-|k|)e^{ikc}\rho(|k|).
\end{align}

\begin{thm} \label{thm-d>1}
Suppose that $\rho(k)$ is summable
(i.e. $\sum_{k=1}^\infty|\rho(k)|<\infty$). Then
for all $c\in(0,2\pi)$,
\[
D_c(n)= S(c) n + e(c,n),
\]
where 
\[
|e(c,n)|\le 2n\sum_{k=n+1}^\infty |\rho(k)|
+2\sum_{k=1}^n k |\rho(k)|=o(n).
\]
In particular $D_c(n)=S(c)n+o(n)$ uniformly in $c$.
\end{thm}

\begin{proof}
Write
\[
D_c(n)=\sum_{k=-n}^n(n-|k|)e^{ikc}\rho(|k|)
=S(c) n+e(c,n),
\]
where
\[
e(c,n)=-n\sum_{k=n+1}^\infty (e^{ikc}+e^{-ikc})\rho(k)
-\sum_{k=1}^n k (e^{ikc}+e^{-ikc})\rho(k).
\]
It remains to show that $\sum_{k=1}^n k|\rho(k)|=o(n)$.
Let $s_n=\sum_{k=1}^n|\rho(k)|$ and let $L=\lim s_n$.
Define the Ces\`aro average
$\sigma_n=\frac1n\sum_{k=1}^n s_k$, so $L=\lim\sigma_n$.
Then $\frac1n\sum_{k=1}^n k|\rho(k)|=\frac{n+1}{n} s_n-\sigma_n\to L-L=0$.
\end{proof}

Under stronger assumptions on the decay rate of the autocorrelation
function $\rho(k)$, improved estimates for the $o(n)$ term are available. 
\begin{thm} \label{thm-d>2}
Suppose that $\sum_{k=1}^\infty k|\rho(k)|<\infty$.  Then
\[
D_c(n)= S(c) n + S_0(c) +e(c,n),
\]
where 
$S_0(c)=-\sum_{k=-\infty}^\infty e^{ic k}|k|\rho(|k|)$ is continuous
on $[0,2\pi]$, $S(c)$ is $C^1$ on $[0,2\pi]$, and 
\[
|e(c,n)|\le 2\sum_{k=n+1}^\infty (k-n)|\rho(k)|=o(1).
\]
In particular, $D_c(n)=S(c)n+O(1)$ uniformly in $c$.
(Hence $K_c\in\{0,1\}$ for all $c\in(0,2\pi)$.)
\end{thm}

\begin{proof}
Compute that
\[
D_c(n)=\sum_{k=-n}^n(n-|k|)e^{ikc}\rho(|k|)
=S(c)\,n+S_0(c)+ e(c,n),
\]
where
$e(c,n)=\sum_{k=n+1}^\infty (k-n) (e^{ikc}+e^{-ikc})\rho(k)$.
\end{proof}

\begin{cor} \label{cor-d>1}
Suppose that $|\rho(k)|\le Ck^{-d}$ for $k\ge1$.
If $1<d<2$, then in Theorem~\ref{thm-d>1},
\[
|e(c,n)| \le 2C\frac{1}{(d-1)(2-d)}n^{2-d}.
\]
If $d>2$, then in Theorem~\ref{thm-d>2},
\[
|e(c,n)| \le 2C\Bigl\{\frac{1}{(d-1)(d-2)}+\frac1n\Bigr\}\,\frac{1}{n^{d-2}}.
\]
\end{cor}

\begin{proof}
The first term of $e(c,n)$ in Theorem~\ref{thm-d>1}
is dominated by
\[
2Cn\sum_{k=n+1}^\infty k^{-d}\le 2Cn\int_n^\infty x^{-d}\,dx \le 
2Cn^{-(d-2)}/(d-1).
\]
The second term is dominated by
\begin{align*}
2C\sum_{k=1}^n k^{1-d} & \le 2C\Bigl(1+\int_1^n x^{1-d}\,dx\Bigr)
=2Cn^{-(d-2)}/(2-d)-
2C(d-1)/(2-d) \\ & \le
2Cn^{-(d-2)}/(2-d).
\end{align*}
Combining these terms gives the result for $1<d<2$.

If $d>2$, then by Theorem~\ref{thm-d>2}, 
\begin{align*}
|e(c,n)| & \le 2C\sum_{k=n+1}^\infty (k-n)k^{-d}
=2C\sum_{k=n+1}^\infty k^{1-d}-2Cn\sum_{k=n+1}^\infty k^{-d} \\
& \le 2C\int_{n}^\infty x^{1-d}\,dx
-2Cn\Bigl(\int_n^\infty x^{-d}\,dx-n^{-d}\Bigr) \\
& =2C\Bigl\{\Bigl(\frac{1}{d-2}-\frac{1}{d-1}\Bigr)\frac{1}{n^{d-2}}+\frac{1}{n^{d-1}}\Bigr\}.
\end{align*}

\vspace{-4ex}
\end{proof}

\section{Nonsummable decay of correlations}
\label{sec-nonsum}

In this section, we 
reformulate the $0$--$1$ test in terms of Ces\`aro averages,
and give surprisingly weak sufficient conditions under which
$S(c)=\lim_{n\to\infty}\frac1n\int_X |p_c(n)|^2d\mu$
exists (for typical values of $c$).

Let $a_k\in\C$ be a sequence with partial sums
$s_n=\sum_{k=-n}^n a_k$ and set $\sigma_n=\frac1n\sum_{k=0}^{n-1} s_k$. 
Recall that the sequence $a_k$ is {\em Ces\`aro summable} if
$\lim_{n\to\infty}\sigma_n$ exists. If $\lim_{n\to\infty}s_n\to L$
then $\lim_{n\to\infty}\sigma_n=L$ (the converse is not true).

Defining $s_n$ and $\sigma_n$ as above with $a_k=e^{ikc}\rho(|k|)$, we obtain
\begin{align*}
\sigma_n=\frac{1}{n}\sum_{k=-n}^{n} (n-|k|)\rho(|k|)e^{ikc}=\frac1n D_c(n),
\end{align*}
the last equality following from~\eqref{eq-D}.
Since $S(c)=\lim_{n\to\infty}\frac1n D_c(n)$, we have proved the following
result.

\begin{lemma} \label{lem-Cesaro}
Let $c\in(0,2\pi)$.   Suppose that the sequence $a_k=e^{ikc}\rho(k)$
is Ces\`aro summable with limit $L(c)$.  Then $S(c)=L(c)$.
In particular, if $L(c)>0$, then $K_c=1$.
\qed
\end{lemma}

By Fej\'er's theorem~\cite{Katzn}, a special case is provided when $\rho(|k|)$
are Fourier coefficients of an integrable function.

\begin{thm}  \label{thm-d>.5}
Suppose that $\rho(|k|)$ are the Fourier coefficients of an $L^1$ function
$g:[0,2\pi]\to\R^{\ge 0}$.  
Then the sequence $e^{ikc}\rho(k)$ is Ces\`aro summable to 
$g(c)$ almost everywhere.  In particular,
Lemma~\ref{lem-Cesaro} holds for almost every $c\in(0,2\pi)$,
with $L(c)=g(c)$.

If $g$ is continuous, then the convergence is uniform in $c$
(and holds for every $c\in(0,2\pi)$).   In particular, $D_c(n)=S(c)n+o(n)$
uniformly in $c$.
\end{thm}

\begin{proof}
We have written $\sigma_n=\frac1n D_c(n)= 
\sum_{k=-n}^n(1-\frac{|k|}{n})e^{ikc}\rho(|k|)$.
This is $\sigma_{n-1}(g,c)$ in~\cite[p.12 (2.9)]{Katzn}).   

If $g$ is continuous, then by Fej\'er's Theorem
(\cite[Theorem~2.12]{Katzn}), $\sigma_n\to g$ uniformly, and
so $D_c(n)=n\sigma_n=ng(c)+o(n)$ uniformly in $c$.

For general $g\in L^1$, it follows from the discussion
in~\cite[pp.~19-20]{Katzn} that $\sigma_n\to g$ almost
everywhere, so that
$D_c(n)=n\sigma_n(g,c)=ng(c)+o(n)$ for almost every $c$.
\end{proof}

\begin{cor} \label{cor-d>.5}
Suppose that $\sum_{k\ge1}\rho(k)^2<\infty$.  Then the first statement
of Theorem~\ref{thm-d>.5} applies with $g(c)=
\sum_{k=-\infty}^\infty e^{ikc}\rho(|k|)$.  

In particular, $S(c)$ exists almost everywhere,
and $D_c(n)=S(c)n+o(n)$.
\end{cor}

\begin{proof}    Since $\sum_{k\ge1}\rho(k)^2<\infty$, the function
$g(c)=\sum_{k=-\infty}^\infty e^{ikc}\rho(|k|)$ lies in $L^2$
and the Fourier coefficients of $g$ are precisely $\rho(|k|)$.
Hence, we can apply the first statement of Theorem~\ref{thm-d>.5}.
\end{proof}

\section{Correlation method}
\label{sec-corr}

In~\cite{GMapp}, we proposed
computing $K_c$ as the correlation of the mean-square
displacement $M_c(n)$ (or $D_c(n)$) with $n$, rather than computing
the limit of $\log M(n)/\log n$ as in~\eqref{eq-K}.
In this section we verify that the theoretical
value of $K_c$ remains $0$ for regular dynamics and $1$ for chaotic dynamics.

Given vectors $x,y$ of length $n$, we define
\begin{align*}
& {\rm cov}(x,y)   = \frac1n\sum_{j=1}^n (x(j)-\bar x)(y(j)-\bar y),
\quad \text{where} \quad
\bar x = \frac1n\sum_{j=1}^n x(j)\; , \\
& \var(x)  =\cov(x,x)\; .
\end{align*}
Form the vectors $\xi=(1,2,\dots,n)$ and
$\Delta=(D_c(1),D_c(2),\dots,D_c(n))$.  
(In particular, $\var(\xi)=\frac{1}{12}(n^2-1)$.)
Define the correlation coefficient
\begin{align*}
K_c = \lim_{n\to\infty}{\rm corr}(\xi,\Delta)=
\frac{{\rm{cov}}(\xi,\Delta)}{\sqrt{{\var(\xi)}{\var(\Delta)}}}\in[-1,1]\;.
\end{align*}

\subsection{Quasiperiodic case}

For quasiperiodic dynamics, we have the following analogue of 
Theorem~\ref{thm-qp}.    However, we require stronger regularity for the
observable $v$, and a Diophantine condition on the frequency $\omega$
(in addition to the condition on $c$).

\begin{thm}
Suppose that $X=\T^m=\R^m/(2\pi\Z)^m$ and that $f:X\to X$ is given by
$f(x)=x+\omega\bmod 2\pi$.  Let $v:X\to\R$ be a nonvanishing $C^r$ 
observable with $r>3m/2$.  If $K_c$
is computed using the correlation method, then for almost every 
$\omega\in[0,2\pi]$ we obtain
$K_c=0$ for almost every $c\in[0,2\pi]$.
\end{thm}

\begin{proof}
The proof of Theorem~\ref{thm-qp} shows that $M_c(n)$
(equivalently $D_c(n)$) is bounded for almost every $c$ provided 
$r>m$.   We show for $r>3m/2$ that 
$\var(\Delta)=O(1)$ and $\cov(\xi,\Delta)=O(1)$
for almost every $\omega$ and $c$.
Moreover, we use the fact that $v$ is nonvanishing
to show that 
$\var(\Delta)=a+O(1/n)$, where $a>0$.
It then follows that
\[
 {\rm corr}(\xi,\Delta)=\frac{O(1)}
{\sqrt{\frac{1}{12}(n^2-1)}\sqrt{a+O(1/n)}}=O(1/n),
\]
as required.

The starting point is the calculations~\eqref{eq-Msn} and~\eqref{eq-formal2}
which gives
\[
M_c(j)= \sum_{\ell}|v_\ell|^2\frac{1-\cos j(c+\ell\cdot\omega)}
{1-\cos(c+\ell\cdot\omega)}=
\sum_{\ell}w_\ell(\cos j\theta_\ell-1),
\]
where $\theta_\ell=c+\ell\cdot\omega$, 
$w_\ell=-|v_\ell|^2(1-\cos \theta_\ell)^{-1}$.
Hence
\[
D_c(j)= \sum_{\ell}w_\ell(\cos j\theta_\ell-1)-C(1-\cos jc),
\]
where $C=(Ev)^2/(1-\cos c)$.
Adding a constant (independent of $j$) to $D_c(j)$
does not alter the value of ${\rm corr}(\xi,\Delta)$ so we may replace
$D_c(j)$ by
\[
\widehat D_c(j)= \sum_{\ell}w_\ell \cos j\theta_\ell+C\cos jc,
\]
when proving that $\cov(\xi,\Delta)=O(1)$ and
$\var(\Delta)=an+O(1)$ with $a>0$.
Hence it suffices to show that for almost every $\omega$ and
$c$ there exists $a>0$ such that 
\[
\sum_{j=1}^n \widehat D_c(j)=O(1),\quad 
\sum_{j=1}^n j\widehat D_c(j)=O(n), \quad  
\sum_{j=1}^n \widehat D_c(j)^2=an+O(1).
\]

Formally,
\begin{align} 
\label{eq-D1}
\sum_{j=1}^n \widehat D_c(j) & = \sum_\ell w_\ell \sum_{j=1}^n\cos j\theta_\ell
+C\sum_{j=1}^n\cos jc, \\
\label{eq-D2}
\sum_{j=1}^n j\widehat D_c(j) & = \sum_\ell w_\ell \sum_{j=1}^n j\cos j\theta_\ell
+C\sum_{j=1}^n j\cos jc, \\
\nonumber
\sum_{j=1}^n \widehat D_c(j)^2 & = 
\sum_{\ell,\ell'} w_\ell w_{\ell'} \sum_{j=1}^n
\cos j\theta_\ell \cos j\theta_{\ell'} \\ & \qquad
+2C\sum_\ell w_\ell \sum_{j=1}^n\cos jc\cos j\theta_\ell 
+C^2\sum_{j=1}^n\cos^2 jc. 
\label{eq-D3}
\end{align}

For $\varphi\in(0,2\pi)$, we have
\begin{align*}
\sum_{j=1}^n e^{ij\varphi} =
e^{i\varphi}\frac{e^{in\varphi}-1}{e^{i\varphi}-1}, \qquad
\sum_{j=1}^n je^{ij\varphi} =
\frac{ne^{i(n+1)\varphi}}{e^{i\varphi}-1}.
\end{align*}
In particular, $\sum_{j=1}^n\cos j\varphi=O(1)$
and $\sum_{j=1}^n j\cos j\varphi=O(n)$.  By~\eqref{eq-D1} and~\eqref{eq-D2},
formally we have $\sum_{j=1}^n \widehat D_c(j)=O(1)$
and $\sum_{j=1}^n j\widehat D_c(j)=O(n)$.   Turning to~\eqref{eq-D3},
\[
\sum_{j=1}^n\cos^2 jc=
\frac12 \sum_{j=1}^n(1+\cos2jc)=\frac12 n +O(1),
\]
for $c\neq\pi$ and formally
\[
\sum_\ell w_\ell \sum_{j=1}^n\cos jc\cos j\theta_\ell
=\frac12\sum_\ell w_\ell \sum_{j=1}^n(\cos j(c+\theta_\ell)+
\cos j(c-\theta_\ell))=O(1),
\]
while
\begin{align*}
& \sum_{\ell,\ell'} w_\ell w_{\ell'} \sum_{j=1}^n
\cos j\theta_\ell \cos j\theta_{\ell'}
 =\frac12\sum_{\ell,\ell'}w_\ell w_{\ell'}\sum_{j=1}^n
\{\cos j(\theta_\ell+\theta_{\ell'}) +\cos j(\theta_\ell-\theta_{\ell'})\} \\
& \enspace =\frac12\sum_{\ell}w_\ell^2 n +
\frac12\sum_{\ell}w_\ell^2\sum_{j=1}^n\cos 2j\theta_\ell 
+\frac12\sum_{\ell\neq\ell'}w_\ell w_{\ell'}\sum_{j=1}^n
\{\cos j(\theta_\ell+\theta_{\ell'}) +\cos j(\theta_\ell-\theta_{\ell'})\} \\
& \enspace  =\frac12\sum_\ell w_\ell^2 n+O(1).
\end{align*}
Hence $\sum_{j=1}^n \widehat D_c(j)^2=an+O(1)$,
with $a=\frac12(\sum_\ell w_\ell^2+C^2)$.
If $v$ is nonvanishing, then $v_\ell$, and hence $w_\ell$,
is nonzero for at least one $\ell$ so that $a>0$.

It remains to justify the formal calculations.   We give the details for the 
first term in~\eqref{eq-D3} focusing on the most difficult expression
\[
I=\sum_{\ell\neq\ell'}w_\ell w_{\ell'}\sum_{j=1}^n
\{\cos j(\theta_\ell+\theta_{\ell'}) +\cos j(\theta_\ell-\theta_{\ell'})\}.
\]
Let $\epsilon>0$.  We assume the Diophantine conditions~\eqref{eq-Dio} and
\begin{align} \label{eq-Dio2}
\dist|2c+\ell\cdot\omega,2\pi\Z|\ge d_2|\ell|^{-(m+\epsilon)}, \quad
\dist|\ell\cdot\omega,2\pi\Z|\ge d_2|\ell|^{-(m+\epsilon)},
\end{align}
which are satisfied by almost all $c$ and $\omega$ for all nonzero $\ell$.
Proceeding as in the proof of Theorem~\ref{thm-qp},
\begin{align*}
|I| & \le C  \sum_{\ell\neq\ell'}|v_\ell|^2 |v_{\ell'}|^2 
(1-\cos\theta_\ell)^{-1} (1-\cos\theta_{\ell'})^{-1}
(|1-e^{i(2c+(\ell+\ell')\omega})|^{-1}
+|1-e^{i(\ell-\ell')\omega}|^{-1}) \\
& \le C'  \sum_{\ell\neq\ell'}|\ell|^{-2r} |\ell'|^{-2r} 
|\ell|^{m+\epsilon} |\ell'|^{m+\epsilon}
(|\ell+\ell'|^{m+\epsilon}+ |\ell-\ell'|^{m+\epsilon}) \\
& \le C''  \sum_{k_1,k_2=1}^\infty k_1^{m-1}k_2^{m-1}k_1^{-2r} k_2^{-2r} 
k_1^{m+\epsilon} k_2^{m+\epsilon} (k_1+k_2)^{m+\epsilon} \\
& \le C'''  \sum_{k_1,k_2=1}^\infty k_1^{m-1}k_2^{m-1}k_1^{-2r} k_2^{-2r} 
k_1^{m+\epsilon} k_2^{m+\epsilon} (k_1^{m+\epsilon}+k_2^{m+\epsilon}) \\
& = 2C'''  \sum_{k_1=1}^\infty k_1^{-(1+2(r-3m/2-\epsilon))}\sum_{k_2=1}^\infty
k_2^{-(1+2(r-m-\epsilon/2))}<\infty
\end{align*}
provided we choose $\epsilon>0$ so small that $r>3m/2+\epsilon$.
 \end{proof}

\subsection{Chaotic case}

Recall that $K_c=1$ in definition~\eqref{eq-K}
if and only if $M_c(n)=an+o(n)$ where $a>0$.
Equivalently $D_c(n)=an+o(n)$ with $a>0$.
We show that this is a sufficient condition for $K_c=1$ via the correlation
method.

\begin{thm}   Let $c\in(0,2\pi)$.
Suppose that $D_c(n)=an+o(n)$ for some $a>0$
and that $K_c$ is computed using the correlation method.  Then $K_c = 1$.
\end{thm}

\begin{proof}   We claim that 
$\cov(\xi,\Delta)=\frac{1}{12}a n^2+o(n^2)$ and
$\var(\Delta)=\frac{1}{12}a^2 n^2+o(n^2)$.  The result is then
immediate.

We verify the claim for $\cov(\xi,\Delta)$.  The verification for
$\var(\Delta)$ is similar.  Write $D_c(n)=an+e(n)$ where $e(n)=o(n)$.
Then
\begin{align*}
\cov(\xi,\Delta) &=\frac1n\sum_{j=1}^n j(aj+e(j))
-\Bigl(\frac1n\sum_{j=1}^n j\Bigr)\Bigl(\frac1n\sum_{j=1}^n(aj+e(j))\Bigr)  \\
& = \frac16 a(n+1)(2n+1)+\frac1n\sum_{j=1}^n je(j)
-\frac{1}{2n}(n+1)\Bigl\{\frac12 an(n+1)+\sum_{j=1}^n e(j)\Bigr\} \\
& = \frac{1}{12}an^2 + O(n) + \frac1n\sum_{j=1}^n je(j)-\frac{1}{2n}(n+1)\sum_{j=1}^n e(j).
\end{align*}
Hence it remains to show that $\sum_{j=1}^n je(j)=o(n^3)$ and
$\sum_{j=1}^n e(j)=o(n^2)$.

Since $e(n)=o(n)$, there is a constant $C>0$ such that $|e(n)|\le Cn$
for $n\ge1$.  Also, given $\epsilon>0$, there exists $n_0\ge1$ such
that $|e(n)|\le\epsilon n$ for all $n\ge n_0$.  Choose $n_1$ such that
$Cn_0^2/n_1^2\le\epsilon$.  Then for all $n\ge n_1$,
\begin{align*}
|\frac{1}{n^2}\sum_{j=1}^n e(j)| & \le \frac{1}{n^2}\Bigl(\sum_{j=1}^{n_0} |e(j)|
+\sum_{j=n_0+1}^n |e(j)|\Bigr)  
 \le \frac{1}{n_1^2}Cn_0^2 + \frac{1}{n^2}\epsilon\sum_{j=1}^n j \\
& \le \epsilon + \epsilon = 2\epsilon,
\end{align*} 
so that $\sum_{j=1}^n e(j)=o(n^2)$.  Similarly $\sum_{j=1}^n je(j)=o(n^3)$.
\end{proof}

\section{Discussion}
\label{sec-potato}
We have addressed the issue of validity of the $0$--$1$ test as
presented in \cite{GM05,GMapp}. The original $0$--$1$ test
\cite{GM04} included an equation for a phase variable $\theta(n+1)
= \theta(n) + c + v\circ f^n$ rather than a constant ``frequency''
$c$. The extra equation for the phase driven by the observable made
available theorems from ergodic theory on skew product systems
\cite{FMT03,MN04,NMA01}.  These theorems rely on the fact that for typical
observables $v$ the augmented system with the phase variable is mixing.

In the modified version proposed in~\cite{GM05}, the phase variable is not
mixing, and the results from ergodic theory are not applicable anymore. 
Nevertheless, the modified test is more effective, particularly for systems
with noise~\cite{GM05}.   In this paper, we have verified that the 
modified test can be rigorously justified.  Moreover, our theoretical
results are stronger than the corresponding results mentioned in~\cite{GM04}
for the original test.

Our main results in this paper are that $K_c=0$ with probability one in the case
of periodic or quasiperiodic dynamics, and that $K_c=1$ with probability one
for ``sufficiently chaotic'' dynamics.   The latter includes
dynamical systems with hyperbolicity, (including weakly mixing
systems such as Pomeau-Manneville intermittency maps).   In particular,
nonuniform hyperbolicity assumptions combined with summable 
autocorrelations for the observable suffice to obtain $K_c=1$.
In the absence of hyperbolicity, we still obtain $K_c=1$
(with probability one) for observables with exponentially decaying
autocorrelations.   These results extend to systems that are mixing
up to a cycle of finite length.

We also made explicit the connection with power spectra: the test
yields $K_c=1$ with probability one if and only if the power spectrum
is well-defined and positive for almost all frequencies.   The criteria
above -- exponential decay of autocorrelations or summable correlations plus
hyperbolicity (up to a finite cycle) -- are sufficient conditions for 
existence and positivity of the power spectrum.

There remains the question of whether typical smooth dynamical systems are
either quasiperiodic or have power spectra that are defined and positive
almost everywhere.   This is required for a complete
justification of the test for chaos.    Unfortunately the current
understanding of dynamical systems is inadequate to answer this question,
but all numerical studies so far indicate this to be the case.
We leave it as a challenge to the skeptical reader to concoct a robust
smooth example where the test fails!
On the positive side, we showed in this paper that under a mild assumption
on autocorrelations, slightly stronger than summable but much weaker
than exponential, we obtain either $K_c=0$ or $K_c=1$ for each choice of $c$,
though without invoking hyperbolicity
we cannot rule out the possibility that both $K_c=0$ and $K_c=1$ occur with 
positive probability.

Our investigations of the validity of the test for chaos enabled us
to construct an improved version of our test. 
The modification which amounts to using $D_c(n)$ rather than the
mean square displacement $M_c(n)$ was shown to significantly improve
the test in \cite{GMapp}.   In addition, we showed in~\cite{GMapp} 
that $K_c$ is better computed by
correlating the mean square displacement with linear growth
rather than computing the $\log$-$\log$ slope.
In this paper, we have shown that our rigorous results apply also
to the improved implementation of the test in~\cite{GMapp}.

\paragraph{Acknowledgements}  
The research of GG was supported in part by the Australian Research Council.
The research of IM was supported in part by EPSRC Grant
EP/F031807/1.

IM acknowledges the hospitality of the University of Sydney where
parts of this work was done, and is greatly indebted to the University
of Houston for the use of e-mail.

\end{document}